\documentclass[a4paper,leqno,11pt]{amsart}

\usepackage{epsf,graphicx,epsfig}
\usepackage{amscd}
\usepackage{amsmath,latexsym,amssymb,amsthm}
\usepackage[nospace,noadjust]{cite}
\usepackage{textcomp}
\usepackage{setspace,cite}
\usepackage{lscape,fancyhdr,fancybox}
\usepackage{xypic}
\usepackage{graphicx}

\usepackage{color}
\usepackage{url}
\usepackage{enumerate}
\usepackage[mathscr]{euscript}

  \theoremstyle{plain}

\swapnumbers
    \newtheorem{thm}{Theorem}[section]
    \newtheorem{prop}[thm]{Proposition}
   \newtheorem{lemma}[thm]{Lemma}

    \newtheorem{subsec}[thm]{}
\theoremstyle{definition}
    \newtheorem{defn}[thm]{Definition}
    \newtheorem{exam}[thm]{Example}

\theoremstyle{remark}
     \newtheorem{remark}[thm]{Remark}

\newcommand{\R}{\mathbb{R}}

\title{}
\author{}
\date{}
\usepackage{amssymb}

\usepackage{hyperref}
\hypersetup{
	colorlinks,
	citecolor=blue,
	filecolor=black,
	linkcolor=blue,
	urlcolor=black
}

\begin{document}

\title{On Generalized Lie Bialgebroids}



\author{Apurba Das}
\email{ apurbadas348@gmail.com}
\address{Stat-Math Unit,
Indian Statistical Institute, Kolkata 700108,
West Bengal, India.}


\thispagestyle{empty}

\begin{abstract} 
An alternative proof of the duality of generalized Lie bialgebroid is given and proved a canonical Jacobi structure can be defined on the base of it. We also introduce the notion of morphism between generalized Lie bialgebroids and proved that the induced Jacobi structure is unique upto a morphism.
\end{abstract}
\maketitle

\section{Introduction} 

The notion of Lie bialgebroid is introduced by Mackenzie and Xu \cite{3} as a generalization of Lie bialgebra and infinitesimal version of Poisson groupoid. Roughly, a Lie bialgebroid is a Lie algebroid $A$ over $M$ such that its dual vector bundle $A^*$ also carries a Lie algebroid structure which is compatible in a certain way with that of $A$. It is shown in \cite{3}, \cite{5} that, if $(A, A^*)$ satisfy the criteria for Lie bialgebroid, then $(A^*, A)$ also satisfy the similar criteria. As an example, if $(M, \pi)$ is a Poisson manifold, then $(TM, T^*M)$ forms a Lie bialgebroid, where $T M$ is the usual tangent Lie algebroid and $ T^*M$ is the Lie algebroid given in the example 2.4 (iii). As a kind of converse, it is also proved that the base space of a Lie bialgebroid carries natural Poisson structure.\\

For any smooth manifold $M$, the vector bundle $ TM \times \mathbb{R} \rightarrow M $ has a natural Lie algebroid structure with $\phi_0 = (0,1) \in \Gamma (T^*M \times \mathbb{R}) = \Omega ^{1} (M) \times \mathbb{R} $ as its 1-cocycle.
Moreover, if $(M, \Lambda, E)$ is a Jacobi manifold, then the 1-jet bundle $T^*M \times \mathbb{R} \rightarrow M $ admits a Lie algebroid structure (example 2.4(iv)) and $X_0 = (-E,0) \in \Gamma (TM \times \mathbb{R})$ is a 1-cocycle of it. In general, the pair $(TM \times \mathbb{R}, T^*M \times \mathbb{R})$ is not a Lie bialgebroid. However if we consider the Lie algebroids $TM \times \mathbb{R}$ and $T^*M \times \mathbb{R}$ together with 1-cocycles $\phi_0 = (0, 1)$ and $X_0 = (-E, 0)$ respectively, then they satisfy some compatibility condition. Motivated from this, Iglesias and Marrero \cite{1} introduced the notion of generalized Lie bialgebroid, so that Jacobi manifold constitutes a generalized Lie bialgebroid. Roughly, a generalized Lie bialgebroid is a pair $((A, \phi_0), (A^*, X_0))$, where $A$ is a Lie algebroid with 1-cocycle $\phi_0 \in \Gamma A^*$, and the dual vector bundle $A^*$ also carries a Lie algebroid structure with $X_0 \in \Gamma A$ as its 1-cocycle and satisfy some compatibility condition in the presence of 1-cocycles. If $\phi_0 = 0$ and $X_0 = 0$, we recover the definition the Lie bialgebroid.
Using the duality result of Lie bialgebroid, it is also proved in \cite{1} that if $((A, \phi_0), (A^*, X_0))$ is a generalized Lie bialgebroid, then $((A^*, X_0), (A, \phi_0))$ is also a generalized Lie bialgebroid (i.e, self-dual). Moreover the base of a generalized Lie bialgebroid carries a natural Jacobi structure.

In this paper, we give a direct proof of the fact that, the concept of generalized Lie bialgebroid is a self-dual without assuming self duality property of Lie bialgebroid and a canonical Jacobi structure is induced on the base of a generalized Lie bialgebroid. We also introduce the notion of morphism between generalized Lie bialgebroids and proved that the induced Jacobi structure is unique upto a morphism.

The paper is organized as follows. In section 2, we recall few defintions and examples. In section 3, we mainly summarize the calculas on lie algebroids in the presence of 1-cocyle. In section 4, we prove the duality of generalized Lie bialgebroid and induced Jacobi structure on the base. Here we also introduce generalized Lie bialgebroid morphism and proved a theorem. Section 5 consists of Triangular Lie bialgebroids.

\section{Preliminaries}

In this section, we recall the definition of Jacobi manifolds \cite{1}, Lie algebroids and Lie bialgebroids (\cite{3}, \cite{4}).

\begin{defn}\label{jacobi manifold}
 Let $M$ be a smooth manifold. A Jacobi structure on M $(2 \leqslant dim M)$ is a bilinear skew symmetric map
$$ \{. , .\} : C^\infty(M) \times C^\infty(M) \rightarrow C^\infty(M) $$
satisfying\\
(i) first order diffrential operator in each argument :
$$\{fg , h \} = f \{g ,h \} + g \{f ,h \} - fg \{1 ,h \}$$
(ii) Jacobi identity :
$$ \{\{f,g\},h\} + \{\{g,h\},f\} + \{\{h,f\},g\} = 0 $$ 
for all $f, g, h \in C^\infty(M).$ A manifold $M$ endowed with such a bracket is called a Jacobi manifold.\\
If $(M, \{ , \})$ is a Jacobi manifold, then from the skew symmetry and property (i) of the bracket, one can associate a bivector field $\Lambda$ and a vector field $E$ on $M$
such that
$$E(\delta f) = \{1, f\}$$
$$ \Lambda (\delta f,\delta g) = \{f,g\} - f E(g) + g E(f) $$

Note that if $E = 0,$ then (M,$\Lambda$) is Poisson manifold \cite{9}.
\end{defn}

\begin{exam}
 (i) Any Poisson manifold is a Jacobi manifold with $E= 0$.\\
(ii) Contact manifolds and l.c.s manifolds are also Jacobi manifolds \cite{2}.
\end{exam}

\begin{defn}\label{lie algebroid}
 A Lie algebroid $(A, [~, ~], \rho)$ over a manifold M is a vector bundle $A$ over
M together with a bundle map $\rho : A \rightarrow T M $ , called the anchor and a Lie
algebra structure $[~, ~]$ on the space $\Gamma(A)$ of the sections of $A,$ such that\\
i) the induced map $ \rho : \Gamma(A) \rightarrow \mathfrak{X}(M) $ is a Lie algebra homomorphism;\\
ii) for any $f \in C^\infty (M)$ and $ X, Y \in \Gamma(A) $, then\\
$$[X, f Y ] = f [X, Y ] + (\rho(X)f )Y.$$
\end{defn}

\begin{exam}

(i)  any Lie algebra can be considered as a Lie algebroid over a point.\\
(ii) For any smooth manifold $M$, its tangent bundle $T M$ is a lie algebroid over $M$ with usual lie bracket on vector fields and the identity map as the anchor map.\\
(iii) Let $(M, \pi)$ be a Poisson manifold and $\pi ^{\sharp} : T^*M \rightarrow T M$ be the bundle map given by $\langle \beta, \pi ^{\sharp} (\alpha) \rangle = \pi (\alpha, \beta)$, then $(T^* (M), [~, ~]_{\pi}, \pi ^{\sharp})$ is a Lie Algebroid over $M$ \cite{9}, where $[~ , ~]_{\pi}$
is the bracket of 1-forms defined by
$$[\alpha, \beta ]_{\pi} = \mathcal{L}_{\pi ^{\sharp} (\alpha)} \beta - \mathcal{L}_{\pi ^{\sharp} (\beta)} \alpha - \delta (\pi(\alpha, \beta))$$
We denote this Lie algebroid by $T^*M |_{\pi}$.\\
(iv) Let $(M, \Lambda, E)$ be a Jacobi manifold, then $T^*M \times \mathbb {R}$ has a Lie algebroid structure $(T^* (M) \times \mathbb {R}, \textlbrackdbl ~ , ~ \textrbrackdbl_{(\Lambda, E)}, \rho_{(\Lambda, E)})$ \cite{1}, where the bracket and anchor is given by\\
$$ \textlbrackdbl (\alpha, f), (\beta, g) \textrbrackdbl_{(\Lambda, E)} = (\mathcal{L}_{\Lambda ^{\sharp} (\alpha)} \beta - \mathcal{L}_{\Lambda ^{\sharp} (\beta)} \alpha - \delta (\Lambda(\alpha, \beta))+ f \mathcal{L}_E {\beta}\\ - g \mathcal{L}_E {\alpha} - \iota_E (\alpha \wedge \beta) ,$$
$$\Lambda(\beta, \alpha) + \Lambda ^{\sharp} (\alpha)(g) - \Lambda ^{\sharp} (\beta) (f) + f E(g) - g E(f)) $$
and $$\rho_{(\Lambda, E)} (\alpha, f) = \Lambda ^{\sharp} (\alpha) + f E $$
for all $(\alpha, f), (\beta, g) \in \Gamma(T^*M \times \mathbb {R}) = \Omega ^1 (M) \times \mathbb {R}.$ When $E= 0$, this reduces to the Lie algebroid on $T^*M$
associated to the poisson manifold $(M, \Lambda)$ defined above. [example (iii)]\\
\end{exam}
  Given a Lie algebroid $(A, [~, ~], \rho)$, the exterior algebra of multisections of
A, $\Gamma (\wedge ^ {\bullet} A)$, together with the generalized
Schouten bracket, forms a Gerstenhaber algebra \cite{4}. Moreover $\Gamma(\wedge^{\bullet} A^*)$ together with the Lie algebroid differential $d$ forms a differential graded algebra, where the differential $d$ has the explicit formula similar to the Cartan differential formula\\
$$ (d \alpha)(X_0, X_1, \ldots, X_n) = \sum_{i=0}^n (-1)^{i} \rho (X_i) \alpha(X_0, .., \hat {X_i},.., X_n)$$
$$ + \sum_{i < j} (-1)^{i+j} \alpha ([X_i,X_j], X_0,.., \hat {X_i},.., \hat {X_j},.., X_n)$$
where $\alpha \in \Gamma(\wedge ^n A^*)$, and $X_0, \ldots, X_n \in \Gamma(A)$. When $A= TM$ is the usual tangent bundle Lie algebroid, denote the differential of the Lie algebroid (i.e, de Rham differential of the manifold) by $\delta$.\\

\begin{defn}\label{lie bialgebroid}
 A Lie bialgebroid is a pair $(A, A^*)$ of Lie algebroids in duality, where the differential $d_{*}$ on $\Gamma (\wedge^{\bullet} A)$ defined by the Lie algebroid structure of $A^*$ and the Gerstenhaber bracket on 
$\Gamma (\wedge^{\bullet} A)$ defined by the Lie algebroid structure of $A$ satisfies\\
$$ d_{*} [X, Y] = [d_{*} X , Y] + [X, d_{*}Y] $$
for all $X, Y \in \Gamma A.$

\end{defn}

\begin{exam}
(i) Any Lie bialgebra is a Lie bialgebroid over a point.\\
(ii) Let $(M, \pi)$ be Poisson manifold, the $(TM, T^{*} M)$ is a lie bialgebroid,
where $T M$ is the usual tangent Lie algebroid and $T^{*} M$ is the Lie algebroid given in the
example 2.4 (iii).\\
(iii) Let $(M, \pi, N)$ is a Poisson-Nijenhuis manifold, that is $\pi$ is a Poisson structure on $M$ and $N$ is a Nijenhuis operator on $M$ which are compatible, then $(TM |_{N}, T^*M |_{\pi})$ forms a Lie bialgebroid, where  $TM |_{N}$ is the Lie algebroid structure on the tangent bundle deformed by the Nijenhuis operator $N$, and $T^*M |_{\pi}$ is the Lie algebroid structure on the cotangent bundle induced from the Poisson structure. (see \cite{6} for more details).

\end{exam}

\section{Cartan calculas on Lie algebroid in the presence of 1-cocycle}

{\bf Identification}:

Let $A \rightarrow M$ be a smooth vector bundle over $M$. Then one can identify the smooth sections of the vector bundle $\wedge^r(A\times \mathbb R) $ with $\Gamma(\wedge^rA) \oplus \Gamma (\wedge^{r-1}A)$, and smooth sections of
$\wedge^k(A^\ast\times \mathbb R)$ with $\Gamma(\wedge^kA^\ast) \oplus \Gamma (\wedge^{k-1}A^\ast)$ such that\\
$$(P, Q)((\alpha_1, f_1), \ldots , (\alpha_r, f_r)) = P(\alpha_1, \ldots, \alpha_r) + \sum_{i=1}^r (-1)^{i+1}f_iQ(\alpha_1, \ldots,\hat{\alpha}_i, \ldots, \alpha_r),$$ 
$$(\alpha, \beta)((X_1, g_1), \ldots , (X_k, g_k)) = \alpha(X_1, \ldots, X_k) + \sum_{j=1}^k (-1)^{j+1}g_j\beta(X_1, \ldots,\hat{X}_j, \ldots, X_k),$$
for $(P, Q) \in \Gamma(\wedge^rA) \oplus \Gamma (\wedge^{r-1}A), (\alpha ,\beta) \in \Gamma(\wedge^kA^\ast) \oplus \Gamma (\wedge^{k-1}A^\ast),$ $(\alpha_i, f_i) \in \Gamma(A^\ast) \oplus C^{\infty}(M, \mathbb R)$ and $(X_j, g_j) \in \Gamma (A)\oplus C^{\infty}(M, \mathbb R)$.

Under the above identifications, the contraction and the exterior product are given by

$$ \iota_{(\alpha, \beta)}(P, Q) = (\iota_{\alpha}P + \iota_{\beta}Q, (-1)^k \iota_{\alpha}Q) $$ 
$$(P, Q) \wedge (P', Q') = (P\wedge P', Q\wedge P' + (-1)^r P\wedge Q'),$$
for $(P', Q') \in \Gamma\wedge^{r'}A \oplus \Gamma{\wedge^{r'-1}A}$.

Now for a Lie algebroid $(A, [~,~], \rho)$ over $M$, a Lie algebroid structure on $A \times \mathbb R$ is described as  
$$[(X,f),(Y,g)] = ([X, Y], \rho(X)g - \rho(Y)f),~~\tilde{\rho}(X,f) = \rho(X),$$ where $X,~Y \in \Gamma A,~~ f,~g \in C^{\infty}(M).$ With the above identifications, the coboundary operator $\tilde{d}$ of the cohomology of $A \times \R $  is given by 
$$\tilde{d} (\alpha, \beta) = (d\alpha, -d\beta),~~ (\alpha, \beta) \in \Gamma {\wedge^k(A^\ast \times \mathbb R)}\cong \Gamma (\wedge^kA^\ast) \oplus \Gamma (\wedge^{k-1}A^\ast).$$\\


{\bf Cartan calculas in presence of 1-cocycle}:

Let $(A, [~,~], \rho)$ be a Lie algebroid over $M$ and let $\phi \in \Gamma A^\ast$ be a $1$-cocycle. Then 
$$\phi([X, Y]) = \rho (X)(\phi(Y)) -\rho(Y) (\phi (X)), ~~ \mbox{for all}~~ X, Y \in \Gamma A.$$
Hence, one has a representation 
$$\rho^{\phi} : \Gamma A \times C^{\infty}(M) \longrightarrow C^{\infty}(M)$$ 
of $A$ on the trivial line bundle twisted by $\phi$ given by
$$\rho^{\phi}(X)f = \rho(X)f + \phi(X)f,~~ \mbox{for}~~ X\in \Gamma A,~ f \in  C^{\infty}(M).$$ 
The coboundary operator of $A$ associated to this representation is denoted by $d^{\phi}$ is called $\phi$- deformed differential and is related to $d$  by the formula 
$$ d^{\phi}\alpha = d \alpha + \phi\wedge \alpha$$

Using the $\phi$-deformed differential and Cartan's formula, one defines the $\phi$-deformed Lie derivative $\mathcal L^{\phi}_X : \Gamma {\wedge^{p} A^{\ast}} \rightarrow \Gamma {\wedge^{p} A^\ast}$ by
$$ \mathcal L^{\phi}_X (\alpha) = d^{\phi}(\iota_X\alpha) + \iota_X(d^{\phi}\alpha), ~~ \mbox{for}~~ X \in \Gamma A,~ \alpha \in  \Gamma {\wedge^pA^{\ast}}.$$

\begin{prop} Then we have:\\
(i) $d^{\phi} (\alpha \wedge \beta) = d^{\phi}\alpha \wedge \beta + (-1)^{|\alpha|} \alpha \wedge d^{\phi}\beta - \phi \wedge \alpha \wedge \beta$\\
(ii) $\mathcal L^{\phi}_X f\alpha = f \mathcal L^{\phi}_X \alpha + \rho(X)f \alpha$\\
(iii) $ \mathcal L^{\phi}_{fX} \alpha = f L^{\phi}_X \alpha + df \wedge \iota_X \alpha $\\
(iv) $\mathcal L^{\phi}_X (\alpha \wedge \beta) = \mathcal L^{\phi}_X (\alpha) \wedge \beta + \alpha \wedge \mathcal L^{\phi}_X (\beta) - \phi (X) \alpha \wedge \beta $
\end{prop}

One also has the $\phi$-deformed Schouten bracket on the space of multisection of $A$, defined by
$$[P,P']^{\phi} = [P, P'] + (r-1) P\wedge \iota_{\phi}P' - (-1)^{r-1}(r'-1)(\iota_{\phi}P)\wedge P',$$ where
$P\in \Gamma{\wedge^rA}$, $P' \in \Gamma{\wedge^{r'}} A.$ 
\begin{thm}(\cite{7})
Let $(A, [~,~], \rho)$ be a Lie algebroid and $\phi \in \Gamma A^\ast$ be a $1$-cocycle. The $\phi$-deformed Schouten bracket satisfies\\
(i) $[X, f]^{\phi} = (\rho^{\phi}(X))f$\\
(ii) $[X, Y]^{\phi} = [X, Y]$\\
(iii) $[P, P']^{\phi} = -(-1)^{(r-1)(r'-1)}[P', P]^{\phi}$\\
(iv) $[P, P'\wedge P'']^{\phi} = [P, P']^{\phi} \wedge P'' + (-1)^{(r-1)r'}P' \wedge [P, P'']^{\phi} - (\iota_{\phi}P)\wedge P' \wedge P''$\\
where $f \in C^{\infty}(M)$, $X,~ Y \in \Gamma A$, $P \in \Gamma {\wedge^r A}$, $P' \in \Gamma {\wedge^{r'}} A$, $P'' \in \Gamma {\wedge^{r''}} A.$
\end{thm}

Then one can define $\phi$-deformed Lie derivative $\mathcal{L}_X^{\phi} : \Gamma{\wedge^r A} \rightarrow \Gamma{\wedge^r A} $ by
$$\mathcal{L}_X^{\phi} P = [X, P]^{\phi} $$
which has the analogous properties of Lie derivaties on multisections twisted by $\phi$.

\section{Generalized Lie Bialgebroid and induced Jacobi structure}

Let $(A, [.,.], \rho)$ be a Lie algebroid over $M$ and $\phi_0 \in \Gamma(A^*)$ a 1-cocyle. Assume that the dual bundle $A^*$ also admits a Lie algebroid
structure $([.,.]_*, \rho_*)$ and $X_0 \in \Gamma(A)$ is its 1-cocyle. Denote the Lie derivative of the Lie algebroid $A$ (resp. $A^*$) is denoted by
$\mathcal{L}$ (resp. $\mathcal{L}_{*}$).
\begin{defn} \cite{1}
The pair $((A, \phi_0), (A^*, X_0))$ is said to be a generalized Lie bialgebroid over $M$ if for all $X, Y \in \Gamma(A)$ and $P \in \Gamma(\wedge^ p A)$
\begin{equation}
  d_* ^{X_0} [X,Y] = [d_* ^{X_0} X , Y]^{\phi_0} + [X , d_* ^{X_0} Y]^{\phi_0} 
\end{equation}
\begin{equation}
\mathcal{L}_{* \phi_0}^{X_0} P + \mathcal{L}^{\phi_0}_{X_0} P = 0 
\end{equation}
\end{defn}

\begin{remark}
(i)  The condition (2) of the above definition is equivalent to
\begin{equation}
 \phi_0 (X_0) = 0,        \rho(X_0) = - \rho_{*} (\phi_0)\\
\end{equation}
and 
\begin{equation}
\mathcal{L}_{* \phi_0} X + [X_0 , X] = 0  
\end{equation}
 
for all $ X \in \Gamma A .$ These follows from (2) by applying $ P = f \in C^\infty(M) $ and $ P = X \in \Gamma A $\\\\
(ii) When $\phi_0 = 0$ and $X_0 = 0,$ we recover the definition of Lie bialgebroid introduced by Mackenzie and Xu \cite{3}.\\

\end{remark}

\begin{exam}
Let $(M, \Lambda, E)$ be a Jacobi manifold, then the pair $((T M \times \mathbb{R}, (0,1)), (T^*M \times \mathbb{R}, (-E,0)))$ is a generalized Lie bialgebroid. \cite{1}\\
Another interesting example of generalized Lie bialgebroid is the one provided by strict Jacobi-Nijenhuis manifolds \cite{8}.\\
\end{exam}

Next we prove the self duality property of generalized Lie bialgebroid. The result will follow from the following sequence of lemmas.\\

 Let $((A, \phi_0), (A^*, X_0))$ be a generalized Lie bialgebroid over $M$, then for all $f, g \in C^\infty (M), X \in \Gamma(A)$ and  $\xi , \eta \in \Gamma(A^*),$ we have
\begin{lemma}\label{1}
\begin{equation}
  d_* ^{X_0} [X,f]^{\phi_0} = [d_* ^{X_0} X , f]^{\phi_0} + [X , d_* ^{X_0} f]^{\phi_0} \\
\end{equation}
\end{lemma}

\begin{proof}
 Since for any arbitrary $ Y \in \Gamma(A) $, we have
$$ d_* ^{X_0} [X,fY] = [d_* ^{X_0} X , fY]^{\phi_0} + [X , d_* ^{X_0} (fY)]^{\phi_0} $$
Expand both sides using the derivation property of Lie algebroids, $\phi_0$-deformed bracket of $A$ and $X_0$-deformed differential of $A^*$. Now using the conditions of the generalized Lie 
bialgebroid (proposition 3.4 in \cite{1}), we have
$$ \mathcal{L} ^{X_0}_{* d^{\phi_0}f} X = [X, d^{X_0}_{*}f] $$
which is same as
$$ \iota_{d^{\phi_0}f} d^{X_0}_{*}X + d^{X_0}_{*} \iota_{d^{\phi_0}f} X = [X, d^{X_0}_{*}f] $$
$$ - [d^{X_0}_{*}X, f]^{\phi_0} + d^{X_0}_{*} \langle d^{\phi_0}f, X \rangle = [X, d^{X_0}_{*}f] $$
$$ d^{X_0}_{*} [X,f]^{\phi_0} = [d^{X_0}_{*}X, f]^{\phi_0} + [X, d^{X_0}_{*}f]^{\phi_0} $$

\end{proof}

\begin{lemma}\label{2}
\begin{equation}
 \mathcal{L} ^{X_0}_{* d ^{\phi_0}f} X + \mathcal{L} ^{\phi_0}_{d^{X_0}_{*}f} X = 0  
\end{equation}
\end{lemma}

\begin{proof}
 since $ d_* ^{X_0} [X,f]^{\phi_0} = [d_* ^{X_0} X , f]^{\phi_0} + [X , d_* ^{X_0} f]^{\phi_0} $\\
we have $[d_* ^{X_0} f , X]^{\phi_0} + d_* ^{X_0} [X,f]^{\phi_0} - [d_* ^{X_0} X , f]^{\phi_0} = 0 $\\
which is equivalent to
$$ [d_* ^{X_0} f , X]^{\phi_0} + d_* ^{X_0} [X,f]^{\phi_0} - [d_* ^{X_0} X , f] + f \iota_{\phi_0} d_* ^{X_0} X = 0 $$
$$ [d_* ^{X_0} f , X]^{\phi_0} + d_* ^{X_0} [X,f]^{\phi_0} + \iota_{df} d_* ^{X_0} X + f \iota_{\phi_0} d_* ^{X_0} X = 0 $$
(since, $[a, f] = - \iota_{df} a$ for any $ a \in \Gamma{\wedge^2A})$
$$ [d_* ^{X_0} f , X]^{\phi_0} + d_* ^{X_0} \iota_{d^{\phi_0} f} X + \iota_{d^{\phi_0} f} d_* ^{X_0} X = 0 $$
$$ \mathcal{L} ^{\phi_0}_{d_{*}^{X_0} f} X + \mathcal{L} ^{X_0}_{*{d ^{\phi_0}}f} X  = 0 $$
\end{proof}

\begin{lemma}\label{3}
\begin{equation}
 \langle d^{X_0}_{*} f, d^{\phi_0} g \rangle + \langle d^{\phi_0} f , d^{X_0}_{*} g \rangle = 0
\end{equation}
\end{lemma}

\begin{proof}
Choose any arbitrary $Y \in \Gamma(A)$, from (6) we have
$$ \mathcal{L} ^{\phi_0}_{d_{*}^{X_0} f} (gY) + \mathcal{L} ^{X_0}_{*{d ^{\phi_0}}f} (gY) - g ( \mathcal{L} ^{\phi_0}_{d_{*}^{X_0} f} Y + \mathcal{L} ^{X_0}_{*{d ^{\phi_0}}f} Y) = 0 $$
therefore, $ [{d_{*}^{X_0} f}, gY]^{\phi_0} + g \mathcal{L} ^{X_0}_{*{d ^{\phi_0}}f} Y + \rho_{*}(d^{\phi_0} f)g Y - g ( \mathcal{L} ^{\phi_0}_{d_{*}^{X_0} f} Y + \mathcal{L} ^{X_0}_{*{d ^{\phi_0}}f} Y ) = 0 $
$$ [{d_{*}^{X_0} f}, gY]^{\phi_0} + g \mathcal{L} ^{X_0}_{*{d ^{\phi_0}}f} Y + \rho_{*}(d^{\phi_0} f)g Y - g ( \mathcal{L} ^{\phi_0}_{d_{*}^{X_0} f} Y + \mathcal{L} ^{X_0}_{*{d ^{\phi_0}}f} Y ) = 0 $$
$$ g [{d_{*}^{X_0} f}, Y] + \rho(d_* ^{X_0} f)g Y + g \mathcal{L} ^{X_0}_{*{d ^{\phi_0}}f} Y + \rho_{*}(d^{\phi_0} f)g Y - g ( \mathcal{L} ^{\phi_0}_{d_{*}^{X_0} f} Y + \mathcal{L} ^{X_0}_{*{d ^{\phi_0}}f} Y) = 0 $$
since $Y$ is arbitrary, this implies
$$ \rho(d_* ^{X_0} f)g + \rho_{*}(d^{\phi_0} f)g = 0 $$
$$ \langle d_* ^{X_0} f , dg \rangle + \langle d^{\phi_0} f , d_{*}g \rangle = 0 $$
$$ \langle d_* ^{X_0} f , d^{\phi_0} g \rangle - \langle d_* ^{X_0} f , g \phi_0 \rangle  + \langle d^{\phi_0} f , d_{*}^{X_0} g \rangle - \langle d^{\phi_0} f , g X_0 \rangle = 0 $$
$$ \langle d_* ^{X_0} f , d^{\phi_0} g \rangle + \langle d^{\phi_0} f , d_{*}^{X_0} g \rangle - g \langle d_{*}f, \phi_0 \rangle - g \langle df, X_0 \rangle = 0 $$
                   $$(\because \langle \phi_0 , X_0 \rangle = 0)$$
that is, $ \langle d_* ^{X_0} f , d^{\phi_0} g \rangle + \langle d^{\phi_0} f , d_{*}^{X_0} g \rangle = 0 ~~ (\because \rho(X_0) = - \rho_{*}(\phi_0))$

\end{proof}

\begin{lemma}\label{4}
\begin{equation}
 \mathcal{L} ^{X_0}_{*{d ^{\phi_0}}f} \xi + \mathcal{L} ^{\phi_0}_{d^{X_0}_{*}f} \xi = 0
\end{equation}
\end{lemma}

\begin{proof}
 since for any function $g \in C^\infty(M)$, we have
\begin{align*}
  \mathcal{L} ^{X_0}_{*{d ^{\phi_0}}f} g + \mathcal{L} ^{\phi_0}_{d^{X_0}_{*}f} g  =& \iota_{d ^{\phi_0}f} d_{*}^{X_0} g + \iota_{d_* ^{X_0} f} d^{\phi_0} g \\
 =& \langle d ^{\phi_0}f , d_* ^{X_0} g \rangle + \langle d_* ^{X_0} f  , d^{\phi_0} g \rangle = 0 
\end{align*}
therefore the result follows from (6) and the identity
$$ (\mathcal{L} ^{X_0}_{*{d ^{\phi_0}}f} + \mathcal{L} ^{\phi_0}_{d^{X_0}_{*}f}) \langle \xi , X \rangle = \langle \mathcal{L} ^{X_0}_{*{d ^{\phi_0}}f} \xi + \mathcal{L} ^{\phi_0}_{d^{X_0}_{*}f} \xi , X \rangle +
\langle \xi , \mathcal{L} ^{X_0}_{*{d ^{\phi_0}}f} X + \mathcal{L} ^{\phi_0}_{d^{X_0}_{*}f} X \rangle $$
\end{proof}

\begin{lemma}\label{5}
\begin{equation}
 d^{\phi_0} [\xi, f]_{*}^{X_0} = [ d^{\phi_0} \xi , f]_{*}^{X_0} + [\xi, d^{\phi_0} f]_{*}^{X_0}
\end{equation}
\end{lemma}

\begin{proof}
since we have $\mathcal{L} ^{X_0}_{*{d ^{\phi_0}}f} \xi + \mathcal{L} ^{\phi_0}_{d^{X_0}_{*}f} \xi = 0$\\
therefore, $ [d^{\phi_0}f, \xi]_{*}^{X_0} + d^{\phi_0} \iota_{d^{X_0}_{*}f} \xi + \iota_{d^{X_0}_{*}f} d^{\phi_0} \xi = 0 $\\
which is equivalent to
$$ [d^{\phi_0}f, \xi]_{*}^{X_0} + d^{\phi_0} [\xi, f]_{*}^{X_0} + \iota_{d_{*}f} d^{\phi_0} \xi + \iota_{fX_0} d^{\phi_0} \xi = 0 $$
$$ [d^{\phi_0}f, \xi]_{*}^{X_0} + d^{\phi_0} [\xi, f]_{*}^{X_0} - [d^{\phi_0} \xi, f]_{*} + f \iota_{X_0} d^{\phi_0} \xi = 0 $$
$$ [d^{\phi_0}f, \xi]_{*}^{X_0} + d^{\phi_0} [\xi, f]_{*}^{X_0} - [d^{\phi_0} \xi, f ]_{*}^{X_0} $$
$$ d^{\phi_0} [\xi, f]_{*}^{X_0} = [d^{\phi_0} \xi, f ]_{*}^{X_0} + [d^{\phi_0} \xi, f]_{*}^{X_0} $$

\end{proof}

\begin{lemma}\label{6}
\begin{equation}
 [\mathcal{L} ^{\phi_0}_{X} , \mathcal{L} ^{X_0}_{* \xi}] (f) - \mathcal{L} ^{X_0}_{* \mathcal{L} ^{\phi_0}_{X} \xi} f + \mathcal{L} ^{\phi_0}_{\mathcal{L} ^{X_0}_{* \xi} X} f = \mathcal{L} ^{\phi_0}_{d_* ^{X_0} \langle \xi , X \rangle} f 
\end{equation} 
\end{lemma}

\begin{proof}
 we first observe that,
\begin{align*}
 \mathcal{L} ^{X_0}_{* \mathcal{L} ^{\phi_0}_{X} \xi} f = \iota_{\mathcal{L} ^{\phi_0}_{X} \xi} d_{*}^{X_0} f =& \langle \mathcal{L} ^{\phi_0}_{X} \phi , d_{*}^{X_0} f \rangle \\
 =& \mathcal{L} ^{\phi_0}_{X} \langle \xi, d_{*}^{X_0} f \rangle - \langle \xi, \mathcal{L} ^{\phi_0}_{X} d_{*}^{X_0} f \rangle\\
 =& \mathcal{L} ^{\phi_0}_{X} \mathcal{L} ^{X_0}_{* \xi} f - \langle \xi, [X, d_{*}^{X_0} f]^{\phi_0} \rangle
\end{align*}
therefore, $\mathcal{L} ^{\phi_0}_{X} \mathcal{L} ^{X_0}_{* \xi} f - \mathcal{L} ^{X_0}_{* \mathcal{L} ^{\phi_0}_{X} \xi} f = \langle \xi, [X, d_{*}^{X_0} f]^{\phi_0} \rangle $\\
similarly one have, $ \mathcal{L} ^{X_0}_{* \xi} \mathcal{L} ^{\phi_0}_{X} f - \mathcal{L} ^{\phi_0}_{\mathcal{L} ^{X_0}_{* \xi} X} f = \langle [ \xi, d^{\phi_0}f]_{*}^{X_0} , X \rangle $\\
 By subtracting we get,
$$ [\mathcal{L} ^{\phi_0}_{X} , \mathcal{L} ^{X_0}_{* \xi}] f - \mathcal{L} ^{X_0}_{* \mathcal{L} ^{\phi_0}_{X} \xi} f + \mathcal{L} ^{\phi_0}_{\mathcal{L} ^{X_0}_{* \xi} X} f = 
 \langle \xi, [X, d_{*}^{X_0} f]^{\phi_0} \rangle - \langle [ \xi, d^{\phi_0}f]_{*}^{X_0} , X \rangle $$
on the other hand,
\begin{align*}
 \mathcal{L} ^{\phi_0}_{d_* ^{X_0} \langle \xi , X \rangle} f = \langle d_* ^{X_0} \langle \xi , X \rangle, d^{\phi_0} f \rangle =& - \langle d_* ^{X_0} f, d^{\phi_0} \langle \xi , X \rangle \rangle \\
 =& - \mathcal{L}^{\phi_0}_{d^{X_0}_{*} f} \langle \xi , X \rangle \\
 =&  - \langle \mathcal{L}^{\phi_0}_{d^{X_0}_{*} f} \xi , X \rangle - \langle \xi, \mathcal{L}^{\phi_0}_{d^{X_0}_{*} f} X \rangle \\
 =& \langle \mathcal{L}^{X_0}_{* d^{\phi_0}f} \xi, X \rangle - \langle \xi, [d^{X_0}_{*} f, X]^{\phi_0} \rangle \\
 =& - \langle [ \xi, d^{\phi_0}f]_{*}^{X_0} , X \rangle + \langle \xi, [X, d^{X_0}_{*} f]^{\phi_0} \\
\end{align*}
hence the result follows.

\end{proof}

Now we are in a position to state one of the main result of this paper:

\begin{prop}
\begin{equation}
 d^{\phi_0} [\xi, \eta ]_{*} = [ d^{\phi_0} \xi , \eta ]_{*}^{X_0} + [\xi, d^{\phi_0} \eta ]_{*}^{X_0}
\end{equation}
\end{prop}

\begin{proof}
 From (10) we have,
\begin{equation}
 ([\mathcal{L} ^{\phi_0}_{X} , \mathcal{L} ^{X_0}_{* \xi}] - \mathcal{L} ^{X_0}_{* \mathcal{L} ^{\phi_0}_{X} \xi} + \mathcal{L} ^{\phi_0}_{\mathcal{L} ^{X_0}_{* \xi} X}) \langle \eta , Y \rangle + 
([\mathcal{L} ^{\phi_0}_{Y} , \mathcal{L} ^{X_0}_{* \eta}] - \mathcal{L} ^{X_0}_{* \mathcal{L} ^{\phi_0}_{Y} \eta} + \mathcal{L} ^{\phi_0}_{\mathcal{L} ^{X_0}_{* \eta} Y}) \langle \xi, X \rangle \\
\end{equation}
 $$= \mathcal{L} ^{\phi_0}_{d_* ^{X_0} \langle \xi , X \rangle} \langle \eta, Y \rangle + \mathcal{L} ^{\phi_0}_{d_* ^{X_0} \langle \eta , Y \rangle} \langle \xi, X \rangle = 0 $$

Now, a direct calculation similar to \cite{2}, shows that\\

$ (\mathcal{L}_{X}^{\phi_0} d_{*}^{X_0} Y - \mathcal{L}_{Y}^{\phi_0} d_{*}^{X_0} X - d_{*}^{X_0} [X,Y] ) (\xi, \eta)
- (\mathcal{L}_{* \xi}^{X_0} d^{\phi_0} \eta - \mathcal{L}_{* \eta}^{X_0} d^{\phi_0} \xi - d^{\phi_0} [\xi, \eta]_{*}) (X, Y)$\\

$ = ([\mathcal{L} ^{\phi_0}_{X} , \mathcal{L} ^{X_0}_{* \xi}] - \mathcal{L} ^{X_0}_{* \mathcal{L} ^{\phi_0}_{X} \xi} + \mathcal{L} ^{\phi_0}_{\mathcal{L} ^{X_0}_{* \xi} X}) \langle \eta , Y \rangle + 
([\mathcal{L} ^{\phi_0}_{Y} , \mathcal{L} ^{X_0}_{* \eta}] - \mathcal{L} ^{X_0}_{* \mathcal{L} ^{\phi_0}_{Y} \eta} + \mathcal{L} ^{\phi_0}_{\mathcal{L} ^{X_0}_{* \eta} Y}) \langle \xi, X \rangle $
$$ - ([\mathcal{L} ^{\phi_0}_{X} , \mathcal{L} ^{X_0}_{* \eta}] - \mathcal{L} ^{X_0}_{* \mathcal{L} ^{\phi_0}_{X} \eta} + \mathcal{L} ^{\phi_0}_{\mathcal{L} ^{X_0}_{* \eta} X}) \langle \xi , Y \rangle - 
([\mathcal{L} ^{\phi_0}_{Y} , \mathcal{L} ^{X_0}_{* \xi}] - \mathcal{L} ^{X_0}_{* \mathcal{L} ^{\phi_0}_{Y} \xi} + \mathcal{L} ^{\phi_0}_{\mathcal{L} ^{X_0}_{* \xi} Y}) \langle \eta, X \rangle$$

since $((A, \phi_0), (A^{*}, X_0))$ is a generalized Lie bialgebroid, therefore the first term of the left hand side vanish, and the right hand side of the above equality is also vanish because of (12).\\
hence  $\mathcal{L}_{* \xi}^{X_0} d^{\phi_0} \eta - \mathcal{L}_{* \eta}^{X_0} d^{\phi_0} \xi - d^{\phi_0} [\xi, \eta]_{*} = 0 $\\
which is equivalent to $ d^{\phi_0} [\xi, \eta ]_{*} = [ d^{\phi_0} \xi , \eta ]_{*}^{X_0} + [\xi, d^{\phi_0} \eta ]_{*}^{X_0} $

\end{proof}

{\bf Observation}:\\
In a generalized Lie bialgebroid, we have
$$ \mathcal{L}_{* \phi_0 } X + [X_0 , X] =0, ~~ \forall  X \in \Gamma A . $$
that is, $\mathcal{L}_{* \phi_0 } X + \mathcal{L}_{X_0} X = 0 $\\
again since we have the identity
\begin{displaymath}
 (\mathcal{L}_{* \phi_0 } + \mathcal{L}_{X_0})\langle \xi, X \rangle = \langle \mathcal{L}_{* \phi_0 } \xi+ \mathcal{L}_{X_0} \xi , X \rangle + \langle \xi , \mathcal{L}_{* \phi_0 } X + \mathcal{L}_{X_0} X \rangle
\end{displaymath}
therefore we can conclude that,
$$ \mathcal{L}_{* \phi_0 } \xi+ \mathcal{L}_{X_0} \xi = 0 $$
(since $ \mathcal{L}_{* \phi_0 } f+ \mathcal{L}_{X_0} f = \rho_{*}(\phi_0) f + \rho (X_0) f = 0 ,$ for any $f \in C^\infty (M) $)\\
hence $ \mathcal{L}_{X_0} \xi + [\phi_0, \xi]_{*} = 0 ,$ for any $\xi \in \Gamma A^* .$\\

From the previous proposition, equation (3) and the observation above, we have the following:

\begin{thm}
 If $((A, \phi_0), (A^*, X_0))$ is a generalized Lie bialgebroid, then $((A^*, X_0), (A, \phi_0))$ is also a generalized Lie bialgebroid.
\end{thm}

 In the next we will prove that, there is a canonical Jacobi structure on the base of a generalized Lie bialgebroid.\\
 Let $((A, \phi_0), (A^*, X_0))$ be a generalized Lie bialgebroid. Define a bracket
 $$\{.,.\}: C^\infty (M) \times C^\infty (M) \rightarrow C^\infty (M)$$
by 
\begin{equation}
 \{f,g\} = \langle d^{\phi_0} f , d^{X_0}_{*} g \rangle
\end{equation}

\begin{prop}
 The bracket defined above satisfies,\\
\begin{equation}
 d^{\phi_0} \{f,g\}= [d^{\phi_0} f, d^{\phi_0} g]_{*}
\end{equation}
\begin{equation}
 d_{*}^ {X_0} \{f,g\} = - [d_{*}^ {X_0}f , d_{*}^ {X_0} g]
\end{equation}

\end{prop}

\begin{proof}
 since $ \{f,g\} = \langle d^{\phi_0} f , d^{X_0}_{*} g \rangle  = \rho^{X_0}_{*} (d^{\phi_0} f) g  = [d^{\phi_0} f , g ]_{*}^{X_0} $
$$\therefore d^{\phi_0} \{f,g\} = [d^{\phi_0} f , d^{\phi_0} g ]_{*}^{X_0} = [d^{\phi_0} f , d^{\phi_0} g ]_{*} $$
similarly, $ \{f,g\} = \rho^{\phi_0} (d_{*}^ {X_0} g) f  = [d_{*}^ {X_0} g , f]^{\phi_0} $
$$ \therefore d_{*}^ {X_0} \{f,g\} = [d_{*}^ {X_0} g , d_{*}^ {X_0} f]^{\phi_0} = - [d_{*}^ {X_0} f , d_{*}^ {X_0} g] $$

\end{proof}

\begin{thm}
  Let $((A, \phi_0), (A^*, X_0))$ be a generalized Lie bialgebroid, then the bracket defined above is a Jacobi structure on the base manifold M.
\end{thm}

\begin{proof}
The bracket is skew-symmetric follows from (7). It is also easy to check that, the bracket satisfies the first order differential operator on each argument.\\
 
To prove the Jacobi identity of the bracket, for any $f \in C^\infty (M),$ define
$$ X_{f} = - d_{*}^{X_0} f \in \Gamma A $$
then $ \rho ^ {\phi_0} (X_{f})g = - \rho ^ {\phi_0} (d_{*}^{X_0} f) g = - \langle d^{\phi_0} g, d_{*}^{X_0} f \rangle = \{ f, g \} $\\
from (15), it follows that,
$$ [X_{f}, X_{g}] = X_{\{f,g\}} $$

Now consider, 
\begin{align*}
 J =&  \{\{f_1,f_2\},f_3\} + \{\{f_2,f_3\},f_1\} + \{\{f_3,f_1\},f_2\} \\
 =& \rho ^ {\phi_0} (X_{\{f_1,f_2\}}) f_3 +  c.p \\
 =& \rho ^ {\phi_0} ([X_{f_1},X_{f_2}]) f_3 +  c.p\\
 =& \rho ^ {\phi_0} (X_{f_1}) \rho ^ {\phi_0} (X_{f_2}) f_3 - \rho ^ {\phi_0} (X_{f_2}) \rho ^ {\phi_0} (X_{f_1}) f_3 + c.p \\
 =& (\{f_1, \{f_2,f_3\}\} - \{f_2, \{f_1, f_3\}) + c.p  \\
 =& 2J
\end{align*}
 (here we have used the identity $\rho ^ {\phi_0} ([X,Y])f = \rho ^ {\phi_0} (X) \rho ^ {\phi_0} (Y)f - \rho ^ {\phi_0} (Y) \rho ^ {\phi_0} (X) f$  which follows directly since $\phi_0 $  is a 1-cocycle)\\
$\therefore J = 0 $, hence the bracket satisfy the Jacobi identity.
\end{proof}

\begin{remark}
 (i) From (13), we have $\{f, g\} = \langle df, d_{*}g \rangle + f \rho_{*}(\phi_0)g + g \rho(X_0)f $\\
 therefore the induced Jacobi bivector field $\Lambda$ and the vector field $E$ is given by
 $$ \Lambda (\delta f, \delta g) = \langle df, d_{*}g \rangle = - \langle dg, d_{*}f \rangle$$ 
 $$E =  \rho_{*}(\phi_0) = - \rho (X_0) $$
 (ii) Let $(M, \Lambda, E)$ be a Jacobi manifold. Consider the generalized Lie bialgebroid $((T M \times \mathbb{R}, (0,1)), (T^*M \times \mathbb{R}, (-E,0)))$ given in Example 4.3, then the induced Jacobi structure on M coincide  with the original Jacobi structure.\\
 (iii) The dual generalized Lie bialgebroid $((A^*, X_0),(A, \phi_0))$ induces the opposite Jacobi structure of the above.
\end{remark}

Next we give the very natural definition of generalized Lie bialgebroid morphism.
\begin{defn}
 A morphism between two generalized Lie bialgebroid $((A, \phi_0), (A^*, X_0))$ and $((B, \psi_0), (B^*, Y_0))$ over $M$ is a map $\Phi : A \rightarrow B $ of Lie algebroids such that the dual map $ \Phi^* : B^* \rightarrow A^* $ is also a Lie algebroid map and they preserves the cocycles. i.e,
$$ \Phi (X_0) = Y_0 , \Phi^* (\psi_0) = \phi_0 $$
\end{defn}

\begin{thm}
 Let $((A, \phi_0), (A^*, X_0))$ be a generalized Lie bialgebroid over a manifold $M$. Then the map
$$ \Phi_A : A \rightarrow TM \times \mathbb{R} $$
defined by, $ \Phi_A (X) = (\rho (X), \phi_0 (X)) $
is a morphism between the generalized Lie bialgebroids $((A, \phi_0), (A^*, X_0))$ and  $((T M \times \mathbb{R}, (0,1)), (T^*M \times \mathbb{R}, (-E,0))),$ where $\rho$ is the anchor of the Lie algebroid $A$.\\
Moreover, if $\Psi : ((A, \phi_0), (A^*, X_0)) \rightarrow ((B, \psi_0), (B^*, Y_0)) $ is a morphism between generalized Lie bialgebroids over $M$, then the corresponding induced Jacobi structures on the base manifold $M$ are same.\\
\end{thm}

\begin{proof}
 The map $\Phi_A$ is clearly a Lie algebroid map and $ \Phi_A (X_0) = (\rho(X_0),0) = (-E,0).$\\
Now the dual map $\Phi_A^{*} : T^{*}M \times \mathbb{R} \rightarrow A^{*}$ is such that
\begin{align*}
  \Phi_A^{*} (\alpha, f) (X) =& \langle (\alpha, f), \Phi_A (X) \rangle \\
 =& \langle (\alpha, f), (\rho(X),\phi_0 (X)\rangle \\
 =& \alpha(\rho(X)) + f \phi_0 (X) 
\end{align*}
therefore, $ \Phi_A ^{*} (\alpha, f) = \rho ^{*}(\alpha) + f \phi_0 $\\
hence $\Phi_A ^{*} (0,1) = \phi_0 $.\\
It is easy to verify that,
 $$ \Phi_A ^{*} : (T^{*}M \times \mathbb{R}, [.,.]_{(\Lambda,E)}, \rho_{(\Lambda, E)}) \rightarrow (A^{*}, [.,.]_{*}, \rho_{*}) $$
preserves the Lie bracket. It also commutes with the anchors, as
\begin{align*}
 \rho_{*} \circ \Phi_A ^{*} (\alpha, f) =& \rho_{*} (\rho ^{*}(\alpha) + f \phi_0) \\
 =& \rho_{*} (\rho ^{*}(\alpha)) + f \rho_{*} (\phi_0) \\
 =& \Lambda^{\sharp} (\alpha) + f E  = \rho_{(\Lambda, E)} (\alpha, f) \\
\end{align*}
To prove the last part of the theorem, 
let the Lie algebroid differential of $A$ and $A^{*}$ (resp. $B$ and $B^{*}$) are denoted by $d_A$ and $d_{A^*}$ (resp. $d_B$ and $d_{B^*}$). Similarly the anchors are denoted by $\rho_A$ and $\rho_{A^*}$ (resp. $\rho_B$ and $\rho_{B^*}$).\\
\begin{align*}
  \{f, g\}_A = \langle d_{A}^{\phi_0} f, d_{A^{*}}^{X_0} g \rangle =& \rho_{A^{*}}^{X_0} (d_{A}^{\phi_0} f) g \\
 =& \rho_{A^{*}}^{X_0} (\Phi_{A}^{*} \tilde \delta ^{(0,1)} f) g \\
 =& \rho_{A^{*}}^{X_0} (\Psi^{*} \Phi_{B}^{*} \tilde \delta ^{(0,1)} f) g \\
 =& \rho_{A^{*}}^{X_0} \Psi^{*} (\Phi_{B}^{*} \tilde \delta ^{(0,1)} f) g \\
 =& \rho_{B^{*}}^{Y_0} (d_{B}^{\psi_0} f) g = \{f,g\}_B \\
\end{align*}
(where  $\tilde \delta ^{(0,1)}$ is $(0,1)$-deformed differential of the Lie algebroid $T M \times \mathbb{R}$ )
\end{proof}
Therefore the induced Jacobi structure on the base of a generalized Lie bialgebroid is unique upto a morphism.

\section{Triangular Generalized lie Bialgebroids}

In this section, we consider a special type of generalized Lie bialgebroids defined by a Lie algebroid $A$ with a 1-cocycle $\phi_0 \in \Gamma A^*$ and a suitable 2- multisection $P \in \Gamma(\wedge^{2}A)$ of it.
This includes triangular generalized Lie bialgebra and the generalized Lie bialgebroid associated to the Jacobi manifolds.\\
  Let $(A, [~,~], \rho)$ be a lie algebroid over $M$ with a 1-cocyle $\phi_0 \in \Gamma A^*$. Let $P \in \Gamma(\wedge^{2}A)$ be a 2-multisection of $A$ such that
$$ [P,P]^{\phi_0} = 0 $$
Define a bracket $[~,~]_{*}$ on $\Gamma A^*$ by
$$ [\alpha, \beta]_{*} := \mathcal{L}^{\phi_0}_{P ^{\sharp} (\alpha)} \beta - \mathcal{L}^{\phi_0}_{P ^{\sharp} (\beta)} \alpha - d^{\phi_0} (P(\alpha, \beta)) $$
for $\alpha, \beta \in \Gamma A^*$.

\begin{thm}
 Let $(A, [~,~], \rho)$ be a Lie algebroid over $M$ with a 1-cocyle $\phi_0 \in \Gamma A^*$. Let $P \in \Gamma(\wedge^{2}A)$ be such that
$ [P,P]^{\phi_0} = 0 $, then the dual bundle $A^{*}$ together with the bracket defined above and the bundle map $\rho_{*} = \rho \circ P^{\sharp} : A^{*} \rightarrow TM $ as anchor is a Lie algebroid. Moreover, $X_0 = - P^{\sharp}(\phi_0) \in \Gamma A$ is a 1-cocycle of it.\\ 
\end{thm}

\begin{proof}
 The bracket is skew-symmetric follows from the expression.\\
Note that,
\begin{align*}
 [\alpha, f \beta]_{*} =& \mathcal{L}^{\phi_0}_{P ^{\sharp} (\alpha)} f \beta - \mathcal{L}^{\phi_0}_{P ^{\sharp} (f \beta)} \alpha - d^{\phi_0} (P(\alpha, f \beta)) \\
 =& f \mathcal{L}^{\phi_0}_{P ^{\sharp} (\alpha)} \beta + (\rho ({P ^{\sharp} \alpha})f) \beta - f \mathcal{L}^{\phi_0}_{P ^{\sharp} (\beta)} \alpha - df \wedge \iota_{P ^{\sharp} (\beta)} \alpha 
 - f d^{\phi_0} (P(\alpha,  \beta)) - df \wedge P(\alpha, \beta) \\
 =& f [\alpha, \beta]_{*} + (\rho ({P ^{\sharp} \alpha})f) \beta \\
\end{align*}
The Jacobi identity of the bracket follows from a direct computation and the condition $ [P,P]^{\phi_0} = 0 $. Therefore $(A^*, [~,~]_{*}, \rho_{*})$ is a Lie algebroid.\\

Note that, here $ P^{\sharp} : A^{*} \rightarrow A $ is a Lie algebroid morphism. Since it commutes with the anchors and for any $\alpha, \beta, \gamma \in \Gamma{A^*}$, we have\\
$ \langle P^{\sharp} ([\alpha,\beta]_{*}), \gamma \rangle $\\
$ = - \langle [\alpha,\beta]_{*} , P^{\sharp} \gamma \rangle $\\
$ = - \langle \mathcal{L}^{\phi_0}_{P ^{\sharp} (\alpha)} \beta , P^{\sharp} \gamma \rangle + \langle \mathcal{L}^{\phi_0}_{P ^{\sharp} (\beta)} \alpha , P^{\sharp} \gamma \rangle + \langle d^{\phi_0} (P(\alpha, \beta)), P^{\sharp} \gamma \rangle $\\
$ = - \mathcal{L}^{\phi_0}_{P ^{\sharp} (\alpha)} \langle \beta, P^{\sharp} \gamma \rangle + \langle \beta, \mathcal{L}^{\phi_0}_{P ^{\sharp} (\alpha)} P^{\sharp} \gamma \rangle +
\mathcal{L}^{\phi_0}_{P ^{\sharp} (\beta)} \langle \alpha, P^{\sharp} \gamma \rangle -
\langle \alpha, \mathcal{L}^{\phi_0}_{P ^{\sharp} (\beta)} P^{\sharp} \gamma \rangle +
\rho ^{\phi_0}(P^{\sharp}\gamma)(P(\alpha, \beta))$\\
$ = - \mathcal{L}^{\phi_0}_{P ^{\sharp} (\alpha)} \langle \beta, P^{\sharp} \gamma \rangle + \langle \beta, [P^{\sharp}\alpha, P^{\sharp}\gamma] \rangle +
\mathcal{L}^{\phi_0}_{P ^{\sharp} (\beta)} \langle \alpha, P^{\sharp} \gamma \rangle - \langle \alpha, [P^{\sharp}\beta, P^{\sharp}\gamma] \rangle + \rho ^{\phi_0}(P^{\sharp}\gamma)(P(\alpha, \beta)) $\\
$ = [P,P]^{\phi_0} (\alpha, \beta, \gamma) + \langle [P^{\sharp} \alpha, P^{\sharp} \beta], \gamma \rangle $\\
$ = \langle [P^{\sharp} \alpha, P^{\sharp} \beta], \gamma \rangle $ \\
therefore, $ P^{\sharp} ([\alpha,\beta]_{*}) = [P^{\sharp} \alpha, P^{\sharp} \beta] $. \\

Now to prove that $X_0 = - P^{\sharp}(\phi_0) $ is a 1-cocycle of the Lie algebroid $A^{*}$, we have to show
$$ X_0 ([\alpha, \beta]_{*}) = \rho_{*}(\alpha) X_0 (\beta) - \rho_{*}(\beta) X_0 (\alpha), \forall \alpha, \beta \in \Gamma A^* $$
Now,
$ \rho_{*}(\alpha) X_0 (\beta) - \rho_{*}(\beta) X_0 (\alpha) - X_0 ([\alpha, \beta]_{*})$\\
$ = - \rho (P^{\sharp} \alpha) \langle P^{\sharp} \phi_0, \beta \rangle + \rho (P^{\sharp} \beta) \langle P^{\sharp} \phi_0, \alpha \rangle  + \langle P^{\sharp} \phi_0 , [\alpha, \beta]_{*} \rangle$\\
$ = - \rho (P^{\sharp} \alpha) (P(\phi_0, \beta)) + \rho (P^{\sharp} \beta) (P(\phi_0, \alpha)) + P(\phi_0, [\alpha, \beta]_{*})$\\
$ = \rho (P^{\sharp} \alpha) \langle \phi_0, P^{\sharp} \beta \rangle - \rho (P^{\sharp} \beta) \langle \phi_0, P^{\sharp} \alpha \rangle - \langle \phi_0 , P^{\sharp} [\alpha, \beta]_{*} \rangle $\\
$ = \rho (P^{\sharp} \alpha) \langle \phi_0, P^{\sharp} \beta \rangle - \rho (P^{\sharp} \beta) \langle \phi_0, P^{\sharp} \alpha \rangle - \langle \phi_0 , [ P^{\sharp} \alpha, P^{\sharp} \beta] \rangle $\\
$ =0$ (since $\phi_0$ is a 1-cocyle of $A$)

\end{proof}

\begin{remark}
 In fact $ P^{\sharp} : (A^{*}, X_0) \rightarrow (A, \phi_0) $ is a morphism between Jacobi algebroids.
$$ (P^{\sharp})^{*}(\phi_0) = X_0 $$
Since $(P^{\sharp})^{*}(\phi_0)(\alpha) = \phi_0 (P^{\sharp} \alpha) = P(\alpha, \phi_0)$\\
on the other hand,\\
 $X_0 (\alpha) = - P^{\sharp} (\phi_0) (\alpha) = - P(\phi_0, \alpha) = P(\alpha, \phi_0) $\\
therefore $ (P^{\sharp})^{*}(\phi_0) = X_0 $. Hence $P^{\sharp}$ is a map between Jacobi algebroids.

\end{remark}

\begin{lemma}
The $X_0$-deformed cohomology of the Lie algebroid $A^*$ is given by
 $$d^{X_0}_{*} X = [P, X]^{\phi_0}, \forall X \in \Gamma{A} $$
\end{lemma}

\begin{proof}
For any $\alpha, \beta \in \Gamma{A^*} $, we have
\begin{align*}
 (d^{X_0}_{*} X)(\alpha, \beta) =& \rho ^{X_0}_{*}(\alpha) X(\beta) - \rho ^{X_0}_{*}(\beta) X(\alpha) - X([\alpha, \beta]_{*})\\
 =& \rho ^{\phi_0} (P ^{\sharp} \alpha) X(\beta) - \rho ^{\phi_0} (P ^{\sharp} \beta) X(\alpha) - \langle X , \mathcal{L}^{\phi_0}_{P ^{\sharp} (\alpha)} \beta - \mathcal{L}^{\phi_0}_{P ^{\sharp} (\beta)} \alpha - d^{\phi_0} (P(\alpha, \beta)) \rangle  \\
 =& \mathcal{L}^{\phi_0}_{P ^{\sharp} (\alpha)} \langle X, \beta \rangle - \mathcal{L}^{\phi_0}_{P ^{\sharp} (\beta)} \langle X, \alpha \rangle - \langle X , \mathcal{L}^{\phi_0}_{P ^{\sharp} (\alpha)} \beta - \mathcal{L}^{\phi_0}_{P ^{\sharp} (\beta)} \alpha - d^{\phi_0} (P(\alpha, \beta)) \rangle \\
 =& \langle \mathcal{L}^{\phi_0}_{P ^{\sharp} (\alpha)} X, \beta \rangle - \langle \mathcal{L}^{\phi_0}_{P ^{\sharp} (\beta)} X, \alpha \rangle + \langle X, d^{\phi_0} (P(\alpha, \beta)) \rangle \\
 =& \langle \mathcal{L}^{\phi_0}_{P ^{\sharp} (\alpha)} X, \beta \rangle - \langle \mathcal{L}^{\phi_0}_{P ^{\sharp} (\beta)} X, \alpha \rangle + \rho ^{\phi_0}(X) (P(\alpha, \beta))  \\
 =& - \langle \mathcal{L}^{\phi_0}_X {P ^{\sharp} (\alpha)} , \beta \rangle + \langle \mathcal{L}^{\phi_0}_X {P ^{\sharp} (\beta)} , \alpha \rangle  + \mathcal{L}^{\phi_0}_X (P(\alpha, \beta)) \\
 =& - \mathcal{L}^{\phi_0}_X \langle {P ^{\sharp} (\alpha)} , \beta \rangle +
\langle P ^{\sharp} (\alpha), \mathcal{L}^{\phi_0}_X \beta \rangle + 
\mathcal{L}^{\phi_0}_X \langle {P ^{\sharp} (\beta)} , \alpha \rangle -
\langle P ^{\sharp} (\beta), \mathcal{L}^{\phi_0}_X \alpha \rangle + \mathcal{L}^{\phi_0}_X (P(\alpha, \beta)) \\
 =& P(\alpha, \mathcal{L}^{\phi_0}_X \beta ) - P (\beta, \mathcal{L}^{\phi_0}_X \alpha ) - \mathcal{L}^{\phi_0}_X (P(\alpha, \beta)) \\
 =& - (\mathcal{L}^{\phi_0}_X P)(\alpha, \beta) = [P, X]^{\phi_0} (\alpha, \beta)
\end{align*}

\end{proof}

\begin{thm}
 The pair $((A, \phi_0), (A^*, X_0))$ is a generalized Lie bialgebroid.
\end{thm}

\begin{proof}
 since $ d^{X_0}_{*} X = [P, X]^{\phi_0} $, for all $X \in \Gamma{A}$. Therefore the compatibility condition (1) of the generalized Lie bialgebroid follows from the graded Jacobi identity of the $\phi_0$-deformed Schouten bracket of $A$.
Conditions (3) and (4) are obvious to check. Hence the proof.

\end{proof}

This type of generalized Lie bialgebroids are called Triangular generalized Lie bialgebroids.\\

\mbox{ }\\

\providecommand{\bysame}{\leavevmode\hbox to3em{\hrulefill}\thinspace}
\providecommand{\MR}{\relax\ifhmode\unskip\space\fi MR }
\providecommand{\MRhref}[2]{%
  \href{http://www.ams.org/mathscinet-getitem?mr=#1}{#2}
}
\providecommand{\href}[2]{#2}

\end{document}